\documentclass[conference]{IEEEtran}
\IEEEoverridecommandlockouts
\def\BibTeX{{\rm B\kern-.05em{\sc i\kern-.025em b}\kern-.08em
    T\kern-.1667em\lower.7ex\hbox{E}\kern-.125emX}}
\usepackage{amsmath,amssymb,mathrsfs,amsthm,bm}
\usepackage{mathtools}
\usepackage{cite}

\usepackage{graphicx,stfloats,subfigure,multicol}
\usepackage{epsfig,epsf,psfrag,amssymb,amsfonts,latexsym,graphicx,mathrsfs,subfigure}
\usepackage[table,xcdraw]{xcolor}
\usepackage{lettrine}
\usepackage{comment}
\usepackage{enumitem}
\newlist{enumsteps}{enumerate}{2}
\setlist[enumsteps,1]{label=Case \arabic*: }
\setlist[enumsteps,2]{label=Case \arabic{enumstepsi}.\arabic*: }
\usepackage{booktabs}
\usepackage[table]{xcolor}
\usepackage{multirow}
\usepackage{tablefootnote}
\usepackage{bbm}
\usepackage{orcidlink}
\usepackage{chngpage}
\usepackage{textcomp}
\usepackage{hyperref}
\usepackage{pifont}
\usepackage{url}
\usepackage[hyphenbreaks]{breakurl}
\usepackage[normalem]{ulem}
\usepackage{algpseudocode}
\newsavebox{\ieeealgbox}

\usepackage{array}

\newtheorem{policy}{Network-aware pricing policy}
\newtheorem{theorem}{Theorem}

\newtheorem{lemma}{Lemma}
\newtheorem{definition}{Definition}

\newtheorem*{policy*}{Dynamic NEM}

\newcommand*{\QED}{\hfill\ensuremath{\square}}

 \def\old#1{}

\def\nn{\nonumber}
\def\beq{\begin{equation}}
\def\eeq{\end{equation}}
\def\bea{\begin{eqnarray}}
\def\eea{\end{eqnarray}}
\def\ba{\begin{array}}
\def\ea{\end{array}}

\def\bitem{\begin{itemize}}
\def\eitem{\end{itemize}}
\def\ben{\begin{enumerate}}
\def\een{\end{enumerate}}

\def\ie{{\it i.e.,\ \/}}

\definecolor{bgrd}{rgb}{1,1,1}
\definecolor{gray}{rgb}{0.5,0.5,0.5}
\definecolor{dkr}{rgb}{0.7,0.1,0.2}
\definecolor{dkb}{rgb}{0.1,0.1,0.8}

\def\tcb{\textcolor{blue}}

\newcommand{\mbbR}{\mathbb{R}}

\def\Bc{{\cal B}}

\def\Dc{{\cal D}}

\def\Kc{{\cal K}}

\def\Nc{{\cal N}}
\def\Oc{{\cal O}}
\def\Pc{{\cal P}}

\def\Xc{{\cal X}}

\begin{document}

\title{Network-Aware and Welfare-Maximizing Dynamic Pricing for Energy Sharing
}

\author{Ahmed S. Alahmed\orcidlink{0000-0002-4715-4379}, Guido Cavraro\orcidlink{0000-0003-0296-720X}, Andrey Bernstein\orcidlink{0000-0003-4489-8388}, and Lang Tong\orcidlink{0000-0003-3322-2681}\thanks{\scriptsize Ahmed S. Alahmed and Lang Tong are with the School of Electrical and Computer Engineering, \textit{Cornell University}, Ithaca, NY, USA ({\tt \{\tcb{ASA278,~LT35}\}\tcb{@cornell.edu}}). Guido Cavraro and Andrey Bernstein are with the Power System Engineering Center, \textit{National Renewable Energy Laboratory}, Golden, CO, USA ({\tt \{\tcb{GUIDO.CAVRARO,~ABERNSTE}\}\tcb{@nrel.gov}}).
}
}

\maketitle

\begin{abstract}
The proliferation of behind-the-meter (BTM) distributed energy resources (DER) within the electrical distribution network presents significant supply and demand flexibilities, but also introduces operational challenges such as voltage spikes and reverse power flows. In response, this paper proposes a network-aware dynamic pricing framework tailored for energy-sharing coalitions that aggregate small, but ubiquitous, BTM DER downstream of a distribution system operator's (DSO) revenue meter that adopts a generic net energy metering (NEM) tariff. By formulating a Stackelberg game between the energy-sharing market leader and its prosumers, we show that the dynamic pricing policy induces the prosumers toward a network-safe operation and decentrally maximizes the energy-sharing social welfare. The dynamic pricing mechanism involves a combination of a locational {\em ex-ante} dynamic price and an {\em ex-post} allocation, both of which are functions of the energy sharing's BTM DER. The {\em ex-post} allocation is proportionate to the price differential between the DSO NEM price and the energy-sharing locational price. Simulation results using real DER data and the IEEE 13-bus test systems illustrate the dynamic nature of network-aware pricing at each bus, and its impact on voltage.
\end{abstract}

\section{Introduction}\label{sec:intro}
\lettrine{W}{hile} BTM DER are primarily adopted to provide prosumer services such as bill savings and backup power, they can also be leveraged, under proper consumer-centric mechanism design, to provide various grid services such as voltage control, system support during contingencies, and new capacity deferrals \cite{Martinietal:22IEEEPEM}. Harnessing the flexibility of BTM DER participation in grid services is usually challenged by the DSO's lack of visibility and controllability on BTM DER alongside the absence of network-aware pricing mechanisms that can induce favorable prosumer behaviors. 

The rising notion of energy sharing of a group of prosumers under the DSO's tariff presents a compelling solution to optimize DER utilization, comply with dynamic grid constraints, and promote renewable energy integration. A major barrier facing the practical implementation of energy-sharing markets is the incorporation of network constraints into their pricing, and aligning the objectives of the self-interested energy-sharing prosumers with the global objective of maximizing the coalition's welfare.


Despite the voluminous literature on energy-sharing systems' DER control and energy pricing, network constraints are rarely considered due to the theoretical complexity they introduce. A short list of recent works on energy communities and energy sharing that neglected network constraints can be found here \cite{Han&Morstyn&McCulloch:19TPS,Alahmed&Tong:24TEMPR, Yang&Guoqiang&Spanos:21TSG,ChenEtal:20TSG}. Some works considered a coarse notion of network constraints by incorporating operating envelopes (OEs) at the point of common coupling between the energy-sharing system and the DSO \cite{Vespermann&Hamacher&Kazempour:21TPS,Fleischhacker&Corinaldesi&Lettner&Auer&Botterud:22TSG} that limit the export and imports between the two entities.
Others considered OEs at the prosumer's level \cite{AzimEtal:24TSG,Alahmed&Cavraro&Bernstein&Tong:23AllertonArXiv}. Few papers considered network-aware pricing mechanisms in distribution networks, such as \cite{Chen&Zhao&Low&Wierman:23TSG, Mediwaththe&Blackhall:21TPS,Li:15CDC} and the line of literature on distribution locational marginal prices (dLMP), e.g. \cite{BaiEtal:18TPS, Papavasiliou:18TSG}.

Our work differs from the existing literature in two important directions. Firstly, we consider network-aware pricing under a generic DSO NEM tariff constraint that charges the energy-sharing platform different prices based on its aggregate net consumption. Secondly, the dynamic network-aware pricing of a platform that is subject to the DSO's fixed and exogenous NEM price gives rise to a market manager's profit/deficit that needs to be re-allocated to the coalition members. We shed light on a unique re-allocation rule that makes the prosumers' payment functions uniform, even if they are located on different buses and the network constraints are binding. Such a re-allocation rule is highly relevant when charging end-users, as it avoids `undue discrimination', which is one of the key principles of rate design outlined by Bonbright \cite{bonbright1988principles}.

In this paper, we present a network-aware and welfare-maximizing pricing policy for energy-sharing coalitions that aggregate DER downstream of a DSO's revenue meter that charges the energy-sharing platform based on a generic NEM tariff. The pricing policy announces an {\em ex-ante} locational, threshold-based, and dynamic price to induce a collective prosumer response that decentrally maximizes the social welfare, while abiding by the network voltage constraints. An {\em ex-post} charge/reward is then used to ensure the market operator's profit neutrality. We show that the market mechanism achieves an equilibrium to the Stackelberg game between the energy-sharing market operator and its prosumers.
Although network constraints couple the decisions of the energy-sharing prosumers, which give rise to locational marginal prices (LMP), we show that by adopting a unique proportional re-allocation rule, the payment function becomes uniform for all prosumers, even if they are located at different buses in the energy-sharing network. Numerical simulations using the IEEE 13-bus test feeder and real BTM DER data shed more light on how the pricing policy influences prosumers' response to ensure safe network operation.

This paper extends our work on Dynamic NEM (D-NEM) without OEs \cite{Alahmed&Tong:24TEMPR} and with OEs \cite{Alahmed&Cavraro&Bernstein&Tong:23AllertonArXiv} by incorporating network constraints, which add substantial complexity, primarily due to coupling of the DER decisions across network buses.

For the column vector $\bm{x}$,  $[\bm{x}]^+$, and $[\bm{x}]^-$ represent its positive and negative elements. The notation $[x]_{\underline{x}}^{\overline{x}}$ represents the projection of $x$ into the closed and convex set $[\underline{x},\overline{x}]$ as per the rule $[x]_{\underline{x}}^{\overline{x}}:=\max\{\underline{x},\min\{x,\overline{x}\}\}$. The notation is also used for vectors, \ie $[\bm{x}]_{\underline{\bm{x}}}^{\overline{\bm{x}}}$.

\section{Proposed Framework and Network Model}\label{sec:framework}
We consider the problem of designing a welfare-maximizing and network-aware pricing policy for an {\em energy sharing} system that bidirectionally transacts energy and money with the DSO under a general NEM tariff. Under NEM, the energy sharing platform imports from the DSO at an {\em import} rate if it is {\em net-consuming}, and collectively exports from the DSO at the {\em export} rate if it is {\em net-producing}. A budget-balanced {\em market operator} is responsible for announcing the market's pricing policy. The market operator uses spatially varying pricing signals to adhere to its network's operational constraints
communicated by the DSO.\footnote{We posit that such DER aggregation schemes are informed by the DSO about their networks' information, including OEs, line thermal limits, voltage limits, among others.} We assume that the timescale of community members’ decision is equivalent to that of the NEM {\em netting period} \cite{Alahmed&Tong:22EIRACM}, which allows us to adopt a single time step formulation.

A radial low voltage distribution network flow model is used to model the network power flow \cite{Baran&Wu:89TPD,Baran&Wu:89TPD2}. Consider a radial distribution network described by $\mathcal{G}=(\Bc,\mathcal{L})$, with $\Bc=\{1,\ldots,B\}$ as the set of {\em energy sharing} buses, excluding bus 0, and $\mathcal{L}=\{(i,j)\}\subset \Bc\times \Bc$ as the set of distribution lines between the buses, with $i,j$ as bus indices. The root bus $0$ represents the secondary of the transformer and is referred to as the slack bus (substation bus). The natural radial network orientation is considered, with each distribution line pointing away from bus $0$. 

For each bus $i\in \Bc$, denote by $\mathcal{L}_i \subseteq \mathcal{L}$ the set of lines on the unique path from buses 0 to $i$, and by $Z_i,q_i$ the active and reactive power consumptions of bus $i$, respectively. The magnitude of the complex voltage at bus $i$ is denoted by $v_i$, and we denote the fixed and known voltage at the slack bus by $v_0$. For each line $(i, j)\in \mathcal{L}$, denote by $r_{ij}$ and $x_{ij}$ its resistance and reactance. For each line, $(i, j)\in \mathcal{L}$, denote by $P_{ij}$ and $Q_{ij}$ the real and reactive power from buses $i$ to $j$, respectively. Let $\ell_{ij}$ denote the squared magnitude of the complex branch current from $i$ to $j$.

We adopt the distribution flow (DistFlow) model, introduced in \cite{Baran&Wu:89TPD}, to model steady state power flow in a radial distribution network, as
\begin{subequations}
\begin{align}
P_{i j}&= -Z_j+\sum_{k:(j, k) \in \mathcal{L}} P_{j k}+r_{i j} \ell_{i j} \label{subeq:activeP}\\
    Q_{i j}&=-q_j+\sum_{k:(j, k) \in \mathcal{L}} Q_{j k}+x_{i j} \ell_{i j} \label{subeq:reactiveP}\\
    v_j^2&=v_i^2-2\left(r_{i j} P_{i j}+x_{i j} Q_{i j}\right)+\left(r_{i j}^2+x_{i j}^2\right) \ell_{i j},\label{subeq:voltage}
\end{align}
\end{subequations}
where $\ell_{ij} =(P^2_{ij}+Q_{ij}^2)/v_i^2$ is the line losses, (\ref{subeq:activeP})-(\ref{subeq:reactiveP}) are the active and reactive power balance equations, and (\ref{subeq:voltage}) is the voltage drop.
We exploit a linear approximation of the DistFlow model above that ignores line losses, given that in practice $\ell_{ij}\approx 0$ for all $(i,j) \in \mathcal{L}$. Therefore, the linearized Distflow (LinDistFlow) equations are given by re-writing (\ref{subeq:activeP})-(\ref{subeq:voltage}) to
\begin{subequations}
\begin{align}
P_{i j}&=-\sum_{k\in \Oc(j)} Z_k,\quad\quad \quad
    Q_{i j}=-\sum_{k\in \Oc(j)}q_k \nn\\
    v_j^2&=v_i^2-2\left(r_{i j} P_{i j}+x_{i j} Q_{i j}\right), \nn
\end{align}
\end{subequations}
where $\Oc(j)$ represent the set node $j$'s descendants, including node $j$, \ie $\Oc(j):=\left\{i:\mathcal{L}_j\subseteq\mathcal{L}_i \right\}$. This gives a solution for $v_i^2$ in terms of $v_0^2$, as
\begin{equation*}
    v_0^2 - v_i^2 = -2\sum_{j\in \Bc}\tilde{R}_{ij}Z_j-2\sum_{j\in \Bc}\tilde{X}_{ij}q_j,
\end{equation*}
where 
\begin{equation}\label{eq:CompactLinDistFlow}
    \tilde{R}_{ij}:=\hspace{-0.4cm}\sum_{(h,k)\in\mathcal{L}_i\cap \mathcal{L}_j}\hspace{-0.2cm}r_{hk},\quad \tilde{X}_{ij}:=\hspace{-0.4cm}\sum_{(h,k)\in\mathcal{L}_i\cap \mathcal{L}_j}\hspace{-0.2cm} x_{hk}
\end{equation}
The LinDistFlow can be compactly written as,
\begin{equation}\label{eq:LinearVoltages}
    \bm{v} = -\bm{R}\bm{Z}-\bm{X}\bm{q}+v_0^2\bm{1},
\end{equation}
where $\bm{v}:=(v_1^2,\ldots,v_B^2), \bm{Z}:=(Z_1,\ldots,Z_B), \bm{q}:=(q_1,\ldots,q_B)$, and
$\bm{R} := [2 \tilde{R}_{ij}]_{B\times B}$ and $\bm{X} := [2 \tilde{X}_{ij}]_{B\times B}$ are the resistance and reactance matrices, respectively. We treat the reactive power $\bm{q}$ as given constants rather than decision variables, which allows us to write (\ref{eq:CompactLinDistFlow}) as
\begin{equation}\label{eq:CompactLinDistFlow2}
     \bm{v} = -\bm{R}\bm{Z}+ \bm{\hat{v}},
\end{equation}
where $\bm{\hat{v}}:= -\bm{X}\bm{q}+v_0^2\bm{1}$. The voltage magnitude vector above is constrained as
\begin{equation}\label{eq:LinDistFlowVolt}
\bm{v}_{\text{min}}\preceq \bm{v} \preceq \bm{v}_{\text{max}},
\end{equation}
where $\bm{v}_{\text{min}}:=v_{\text{min}}^2 \bm{1}$ and $\bm{v}_{\text{max}}:=v_{\text{max}}^2 \bm{1}$. Given that the second term in (\ref{eq:CompactLinDistFlow2}) is fixed, we re-write (\ref{eq:LinDistFlowVolt}) to
\begin{equation}\label{eq:LinDistFlowVolt2}
    \bm{\underline{v}}\preceq \bm{v} \preceq \bm{\overline{v}},
\end{equation}
where $\overline{\bm{v}}:=\bm{v}_{\text{max}}-\bm{\hat{v}}$ and $\underline{\bm{v}}:= \bm{\hat{v}}-\bm{v}_{\text{min}}$. We will impose (\ref{eq:LinDistFlowVolt2}) on the operation of the energy-sharing market.

\section{Energy Sharing Mathematical Model}\label{sec:model}
Let $\Nc:=\{1,\ldots,N\}$ denote the set of {\em energy-sharing} system's prosumers. Every prosumer $n$ is connected to one of the $B$ buses in the considered radial network through its revenue meter that measures the prosumer's net consumption and BTM generation. Figure \ref{fig:NetworkEC} shows an example 4-bus energy-sharing platform. We denote the set of prosumers connected to bus $i\in \Bc$ by $\Nc_i$, hence, $\Nc =\bigcup_{i\in \Bc} \Nc_i$. In this section, we model prosumers' DER in $\S$\ref{subsec:DERmodel}, and payment and surplus functions in $\S$\ref{subsec:PaymentSurplus}, followed by a formulation of the proposed bi-level program in $\S$\ref{subsec:StackGame}.

\begin{figure}
    \centering
    \includegraphics[scale=0.57]{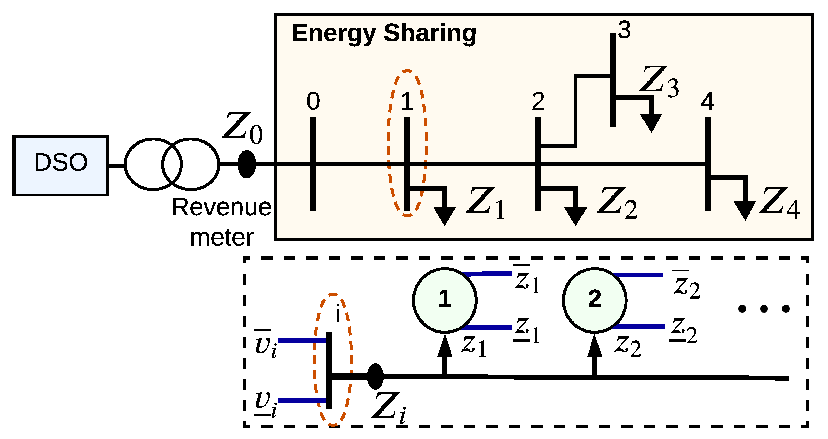}
    \vspace{-0.3cm}
    \caption{A 4-bus energy-sharing platform. $Z_0,Z_i, z_{n} \in \mbbR$ are the net consumption of the whole energy sharing platform, net consumption of bus $i$, and net consumption of prosumer $n$, respectively. $\overline{z}_n\geq 0$ and $\underline{z}_n \leq0$ are the prosumer's import and export OEs, respectively.}
    \label{fig:NetworkEC}
\end{figure}

\subsection{DER Modeling}\label{subsec:DERmodel}
Prosumers' DER consists of BTM renewable distributed generation (DG) and flexible loads (decision variables). The random {\em renewable DG} output of every prosumer $n\in \Nc$ is denoted by $g_n \in \mbbR_+$. The vector of prosumers' DG output is denoted by $\bm{g}:=(g_1,\ldots,g_N)$, and the aggregate DG output in the energy-sharing platform is defined by $G_0 = \sum_{n\in \Nc} g_n$.

The {\em flexible loads}' energy consumption vector is denoted by $\bm{d}_n \in \mbbR_+^K$, where $\Kc:=\{1,\ldots,K\}$ is the load (device) bundle. The devices are subject to their flexibility limits, as
\begin{equation}\label{eq:Conslimit}
 \bm{d}_n \in \Dc_n := [\underline{\bm{d}}_n, \overline{\bm{d}}_n],~~~ \forall n\in \Nc,
\end{equation}
where $ \underline{\bm{d}}_n$ and $\overline{\bm{d}}_n$ are the device bundle's lower and upper consumption limits of prosumer $n\in \Nc$, respectively. 

The {\em net consumption} $z_n\in \mbbR$ of each prosumer is the difference between its gross consumption and BTM generation, hence $z_n = \bm{1}^\top  \bm{d}_n - g_n$.\footnote{The proposed pricing policy can be generalized to incorporate OEs with only little mathematical complication. We show this in the appendix.} The {\em aggregate energy-sharing net consumption} is therefore $Z_0 = \sum_{n\in \Nc} z_{n}= \sum_{i\in \Bc} Z_i$.

\subsection{Payment, Surplus, and Profit Neutrality}\label{subsec:PaymentSurplus}

The energy sharing operator designs a pricing policy $\chi$ for its members, which specifies the payment function for each prosumer $n\in \Nc$ under $\chi$, denoted by $C^\chi_n(z_n)$. Energy-sharing prosumers are assumed to be rational and self-interested. Therefore, they schedule their DER based on {\em surplus maximization}. {\em Prosumer surplus} is given by
\begin{equation}\label{eq:Surplusi}
    S^{\chi}_n(\bm{d}_n, g_n):= U_n(\bm{d}_n)-C^{\chi}_n(z_n),~~ z_n=\bm{1}^\top \bm{d}_n - g_n,
\end{equation}
where for every $n\in \Nc$, the {\em utility function} of the consumption bundle $U_n(\bm{d}_n)$ is assumed to be non-decreasing, additive, strictly concave, and continuously differentiable with a {\em marginal utility function} $\bm{L}_n:=\nabla U_n=\left(L_{n 1}, \ldots, L_{n K}\right)$. Denote the {\em inverse marginal utility} vector by $\bm{f}_n:=(f_{n1},\ldots,f_{nK})$ with $f_{nk}:= L^{-1}_{nk}, \forall n\in \Nc, k\in \Kc$.  

\subsubsection{Energy Sharing Payment}\label{subsec:operatorP}
The operator transacts with the DSO under the {\em NEM X tariff}, introduced in \cite{Alahmed&Tong:22IEEETSG}, which charges the energy sharing coalition based on whether it is {\em net-importing} ($Z_0>0$) or {\em net-exporting} ($Z_0<0$) as
 \begin{align}
       \pi^{\mbox{\tiny NEM}}(Z_0) = \begin{cases}
\pi^+, &\hspace{-0.85em} Z_0\geq 0 \\ 
\pi^-, &\hspace{-0.85em} Z_0 < 0
\end{cases},
    C^{\mbox{\tiny NEM}}(Z_0) = \pi^{\mbox{\tiny NEM}}(Z_0)\cdot Z_0,\label{eq:Pcommunity}
    \end{align}
where $(\pi^+, \pi^-) \in \mbbR_+$ are the {\em buy} (retail) and {\em sell} (export) rates, respectively. 
We assume $\pi^+ \geq \pi^-$, in accordance with NEM practice \cite{Alahmed&Tong:22EIRACM}, which also eliminates risk-free price arbitrage. The operator of the energy sharing regime is {\em profit-neutral}; a term we define next.
\begin{definition}[Profit neutrality]\label{ax:ProfitNeutrality}
The operator is profit-neutral if its pricing achieves the following
\[\sum_{n\in \Nc} C^\chi_n(z_n) = C^{\mbox{\tiny NEM}}(\sum_{n\in \Nc} z_n).\]
\end{definition}
The challenging question we ask is how can the operator design a payment $C^\chi_n$, for every $n\in \Nc$, to achieve network-awareness, profit neutrality and equilibrium to the energy-sharing market, which we define next.

\subsection{Energy Sharing Stackelberg Game}\label{subsec:StackGame}
We formulate this game as a bi-level mathematical program with the upper-level optimization being the operator's pricing problem, and the lower-level optimizations representing prosumers' optimal decisions. 

Denote the consumption policy of the $n$th prosumer, given the pricing policy $\chi$, by $\psi_{n,\chi}$.
Formally, 
$$\psi_{n,\chi}: \mathbb{R}_+ \rightarrow \Dc_n, \ g_n \stackrel{C^\chi_n}{\mapsto} \psi_{n,\chi}(g_n),$$
with $\bm{\psi}_{\chi}:=\{\psi_{1,\chi},\ldots,\psi_{N,\chi}\}$ as the vector of prosumers' policies. The operator strives to design a network-aware and welfare-maximizing pricing policy $\chi^{\sharp}_{\bm{\psi}}$  (given $\bm{\psi}$), where $\chi^{\sharp}_{\bm{\psi}}: \mathbb R_+^N \rightarrow \mathbb R^N, \ \bm{g} \mapsto \bm{C}^{\chi}:=(C^\chi_1,\ldots, C^\chi_N),$ and the welfare is defined as the sum of total prosumers' surplus, as
\begin{equation*}
W^{\chi,\bm{\psi}_{\chi}}:=\sum_{n\in \Nc} S_n^{\chi}(\psi_{n,\chi}(g_n),g_n),
\end{equation*}
The bi-level program can be compactly formulated as
\begin{subequations}\label{eq:community_optimization}
\begin{align}
 \underset{\bm{C}(\cdot)}{\operatorname{maximize}}& \Bigg(W^{\chi_{\bm{\psi}}} = \sum_{n\in \Nc} U_n(\bm{d}^{\psi^\sharp_\chi}_{n}) - C^{\mbox{\tiny NEM}}(Z_0^{\psi^\sharp_\chi})\Bigg)\label{eq:UpperLevelObj}\\
 \text{subject to}&\quad \quad \sum_{n\in \Nc} C^\chi_n(z_n^{\psi^\sharp_\chi}) = C^{\mbox{\tiny NEM}}(Z_0^{\psi^\sharp_\chi})\label{eq:UpperLevelProfit}\\&
 \quad \quad Z_0^{\psi^\sharp_\chi} = \sum_{n\in \Nc} \big(\bm{1}^\top \bm{d}^{\psi^\sharp_\chi}_{n}- g_n\big)\label{eq:UpperLevelEnergyBalance}\\
 (\underline{\bm{\eta}},\overline{\bm{\eta}})&\quad \quad \underline{\bm{v}} \preceq  -\bm{R} \bm{Z}^{\psi^\sharp_\chi} \preceq \overline{\bm{v}}\label{eq:UpperLevelNetwork}\\
 &\quad \quad \text{for all } n \in \Nc \\
 &	\bm{d}^{\psi^\sharp_\chi}_{n}\hspace{-0.1cm} := \underset{\bm{d}_n \in \Dc_n}{\operatorname{argmax}}~S^\chi_n(\bm{d}_n,g_n)\hspace{-0.1cm}:=U_n(\bm{d}_n)-C^\chi_n(z_n)\label{eq:LowerLevelObj}\\
 &	\hspace{1.3cm}	\text{subject to}~~ z_n= \bm{1}^\top \bm{d}_n - g_n,\label{eq:LowerLevelEnergyBalance}
\end{align}
\end{subequations}
where 
$$\bm{Z}^{\psi^\sharp_\chi} :=(\sum_{n\in \Nc_1}\bm{1}^\top \bm{d}^{\psi^\sharp_\chi}_{n}-g_n ,\ldots, \sum_{n\in \Nc_B}\bm{1}^\top \bm{d}^{\psi^\sharp_\chi}_{n}-g_n).$$
In the following, we will assume that problem~\eqref{eq:community_optimization} is feasible, \ie a solution meeting all the constraints exists.

The program in (\ref{eq:community_optimization}) defines the Stackelberg strategy.  
Specifically, ($\chi^\ast,\bm{\psi}^\ast$) is a Stackelberg equilibrium since (a) for all $\chi \in \Xc$ and $n\in \Nc$,  $ S^\chi_n(\psi^\ast_n (g_n),g_n) \ge S^\chi_n(\psi_n (g_n),g_n)$ for all $\bm{\psi} \in \Psi$; (b) for all $\bm{\psi} \in \Psi, W^{\chi^\ast,\bm{\psi}^\ast} \ge \sum_n S_n^\chi(\psi^\ast_n(g_n),g_n)$. 

\section{Network-Aware Pricing and Equilibrium}\label{sec:MktMech}
At the beginning of each pricing period, the operator communicates the price to each prosumer. Given the price, prosumers simultaneously move to solve their own surplus maximization problem. At the end of the netting period, and given the resulting $Z_0$, the DSO charges the energy sharing operator based on the NEM X tariff in (\ref{eq:Pcommunity}). We propose the network-aware pricing policy and delineate its structure in $\S$\ref{subsec:NetAwaPrice}, followed by solving the optimal response of prosumers in $\S$\ref{subsec:ProsOptDec}. We discuss the operator's profit/deficit redistribution in $\S$\ref{subsec:ExpostAllocation} and $\S$\ref{subsec:PriceUniformity}. In $\S$\ref{subsec:EquilibriumThm}, we establish the market equilibrium result. 

\subsection{Network-Aware Dynamic Pricing}\label{subsec:NetAwaPrice}
The operator uses the renewable DG vector $\bm{g}$ to dynamically set the price taking into account network constraints. That is, the dynamic price is used to satisfy network constraints in a decentralized way by internalizing them into prosumers' private decisions.

\begin{policy}
For every bus $i\in \Bc$, the operator charges the prosumers based on a two-part pricing
\begin{equation}\label{eq:PricingPolicy}
    \chi^{\ast}: \bm{g} \mapsto C^{\chi^\ast}_n(z_n)=\underbrace{\pi^\ast_i(\bm{g})}_\text{ex-ante price}\cdot z_n - \underbrace{A_n^\ast}_\text{ex-post allocation}, \forall n\in \Nc_i,
\end{equation}
where the {\em ex-ante} bus price $\pi^\ast_i(\bm{g})$ abides by a two-threshold policy with thresholds
\begin{equation}\label{eq:thresholds}
\begin{aligned}
    \sigma_1(\bm{g}) &= \sum_{i\in \Bc} \sum_{n \in \Nc_i} \bm{1}^{\top} {[\bm{f}_n(\bm{1}\chi_i^+(\bm{g}))]}_{\underline{\bm{d}}_n}^{ \overline{\bm{d}}_n},\\ \sigma_2(\bm{g}) &= \sum_{i\in \Bc} \sum_{n \in \Nc_i} \bm{1}^{\top} {[\bm{f}_n(\bm{1}\chi_i^-(\bm{g}))]}_{\underline{\bm{d}}_n}^{ \overline{\bm{d}}_n}\geq \sigma_1(\bm{g}),
\end{aligned}
\end{equation}
as
\begin{equation}\label{eq:BusPrice}
    \pi^\ast_i(\bm{g}) = \begin{cases} \chi^+_i(\bm{g}) & , G_0< \sigma_1(\bm{g}) \\ \chi^z_i(\bm{g}) & , G_0\in [\sigma_1(\bm{g}),\sigma_2(\bm{g})]\\ \chi^-_i(\bm{g}) & , G_0> \sigma_2(\bm{g}), \end{cases}\quad
\end{equation}
and the price $\chi_i^{\kappa}$, where $\kappa:=\{+,-,z\}$, is given by
\begin{equation}\label{eq:chiKappa}
    \chi_i^{\kappa} = \pi^\kappa-\sum_{j=1}^B R_{ji} (\overline{\eta}^\ast_j-\underline{\eta}^\ast_j)
\end{equation}
where $\overline{\eta}^\ast_j$ and $\underline{\eta}^\ast_j$ are the dual variables of the upper and lower voltage limits in (\ref{eq:UpperLevelNetwork}), respectively, and the price $\pi^z:=\mu^\ast$ is the solution of
\begin{equation}\label{eq:NZprice}
    \sum_{i\in \Bc} \sum_{n\in \Nc_i} \bm{1}^\top [\bm{f}_n(\bm{1}\mu-\bm{1}\sum_{j=1}^B R_{ji}(\overline{\eta}^\ast_j-\underline{\eta}^\ast_j))]_{\underline{\bm{d}}_n}^{\overline{\bm{d}}_n}= G_0.
\end{equation}

The two pricing policy parts are composed of a locational dynamic price that is announced {\em ex-ante} and a charge (reward) that is distributed {\em ex-post}. For every bus $i\in \Bc$, the prosumer's ex-post charge/reward is denoted by $ A_n^\ast$, which we delineate in $\S$\ref{subsec:ExpostAllocation} and $\S$\ref{subsec:PriceUniformity}.
\end{policy}

 The locational {\em ex-ante} price $\pi^\ast_i(\bm{g})$ for every $i\in \Bc$ is used to induce a collective prosumer response at each bus so that the network constraints are satisfied and the energy sharing social welfare is maximized. The energy-sharing price has a similar structure to the celebrated LMP in wholesale markets \cite{Schweppe&etal:88Springer} in the sense that it takes into account demand, generation, location, and network physical limits. Also, like {\em congestionless LMP}, the energy-sharing price is uniform across all buses if the network constraints are nonbinding, as described in (\ref{eq:chiKappa}). 

 Similar to D-NEM without network constraints \cite{Alahmed&Tong:24TEMPR}, the price obeys a two-threshold policy and it is a monotonically decreasing function of the system's renewables $\bm{g}$. As shown in (\ref{eq:chiKappa}), the thresholds partition $G_0$, and the price at each bus is the D-NEM price adjusted by the shadow prices of violating voltage limits. When $G_0 \in [\sigma_1(\bm{g}),\sigma_2(\bm{g})]$ the platform is energy-balanced, and the price $\chi^z_i(\bm{g})$ is the sum of the dual variables for energy balance and voltage limits. 

The thresholds and locational prices can be computed while preserving prosumers' privacy. The operator do not need the functional form of prosumers' utilities or marginal utilities but rather asks the prosumers to submit a value for every device $k$ at a given price.

\subsection{Optimal Prosumer Decisions}\label{subsec:ProsOptDec}
After the network-aware price is announced, consumers simultaneously move to solve their own surplus maximization problem by determining their optimal decision policy $\psi^\ast_{n,\chi^\ast}: \mbbR_+ \rightarrow \Dc_n, \ g_n \stackrel{{C^{\chi^\ast}_n}}{\mapsto} \bm{d}_n^{\psi^\ast}:= \psi^\ast_n(g_n), \forall n\in \Nc$. Therefore, from the surplus definition in (\ref{eq:Surplusi}), each prosumer solves
\begin{align} 
\bm{d}_n^{\psi^\ast}=& \underset{\bm{d}_n \in \Dc_n}{\operatorname{argmax}}~~  S_n^{\chi^\ast}(\bm{d}_n,g_n):=U_n\left(\bm{d}_n\right)-  \pi^\ast_i(\bm{g}) \cdot z_n\nn\\
& \text {subject to } \hspace{0.2cm} z_n = \bm{1}^\top \bm{d}_n-r_n, \label{eq:argmax}
\end{align} 
where $A^\ast_n$ was omitted because it is announced after consumption decisions are exercised. 

\begin{lemma}[Prosumer optimal consumption]\label{lem:OptSchedule}
    Under every bus $i\in \Bc$, given the pricing policy $\chi^\ast$, the prosumer's optimal consumption is
    \begin{align}\label{eq:MemberOptz}
        \bm{d}^{\psi^\ast}_{n}(\pi^\ast_i) = [\bm{f}_{n}(\bm{1}\pi^\ast_i)]_{\underline{\bm{d}}_{n}}^{\overline{\bm{d}}_{n}},~~ \forall n\in \Nc_i.
    \end{align}
    By definition, the aggregate net consumption is 
\begin{align}\label{eq:OptNetCons}
        Z_0^{\psi^\ast}\hspace{-0.15cm}(\bm{\pi}^\ast(\bm{g})) \hspace{-0.1cm}=\hspace{-0.1cm} \sum_{i\in \Bc} \sum_{n\in \Nc_i}\hspace{-0.15cm}\begin{cases}
\bm{d}^{\psi^\ast}_{n}(\chi^+_i) - g_n &\hspace{-0.3cm}, G_0 < \sigma_1(\bm{g}) \\ 
\bm{d}^{\psi^\ast}_{n}(\chi^z_i) - g_n &\hspace{-0.3cm}, G_0 \in [\sigma_1(\bm{g}),\sigma_2(\bm{g})]\\ 
\bm{d}^{\psi^\ast}_{n}(\chi^-_i) - g_n &\hspace{-0.3cm} ,G_0 > \sigma_2(\bm{g}),
\end{cases}
\end{align}
where $\bm{\pi}^\ast:=(\pi^\ast_1,\ldots,\pi^\ast_B)$, and $Z_0^{\psi^\ast}(\bm{\pi}^\ast(\bm{g})) > 0$ if $G_0 < \sigma_1(\bm{g})$, $Z_0^{\psi^\ast}(\bm{\pi}^\ast(\bm{g})) = 0$ if $G_0 \in [\sigma_1(\bm{g}),\sigma_2(\bm{g})]$, and $Z_0^{\psi^\ast}(\bm{\pi}^\ast(\bm{g})) < 0$ if $G_0 > \sigma_2(\bm{g})$.
\end{lemma}
\begin{proof} We drop the prosumer subscript $n$ for brevity. The objective in (\ref{eq:argmax}) is strictly concave and differentiable. The
Lagrangian function of the surplus maximization problem, for a prosumer under bus $i$, is
$$\mathscr{L}(\bm{d},\overline{\bm{\gamma}},\underline{\bm{\gamma}})= \pi^\ast_i(\bm{g}) \cdot z - U\left(\bm{d}\right) + \overline{\bm{\gamma}}^\top (\bm{d} - \overline{\bm{d}}) - \underline{\bm{\gamma}}^\top (\bm{d} - \underline{\bm{d}}),$$
where $\overline{\bm{\gamma}}\in \mbbR_+^K$ and $\underline{\bm{\gamma}}\in \mbbR_+^K$ are the Lagrangian multipliers of the upper and lower consumption limits. From the KKT conditions we have 
$$\nabla_{\bm{d}} \mathscr{L} = \bm{1}\pi^\ast_i(\bm{g}) - \bm{L}(\bm{d}^{\psi^\ast}) + \overline{\bm{\gamma}} - \underline{\bm{\gamma}} = \bm{0},$$ 
therefore, for each device $k\in \Kc$, we have
\begin{align}
    d_k^{\psi^\ast} &= \begin{cases}
f_k(\pi^\ast_i) &, \overline{\gamma}_k = \underline{\gamma}_k = 0 \\ 
\overline{d}_k &, \overline{\gamma}_k>0,  \underline{\gamma}_k = 0\\ 
\underline{d}_k & ,\overline{\gamma}_k=0,  \underline{\gamma}_k > 0
\end{cases}\nn\\
&=:  [f_k(\pi^\ast_i)]_{\underline{d}_k}^{\overline{d}_k},\nn
\end{align}
where $f_k := L^{-1}_k$.

Give the aggregate net consumption definition $Z_0 = \sum_{n\in \Nc} (\bm{1}^\top \bm{d}_{n} - g_{n})$ and the dynamic price in (\ref{eq:BusPrice}), one can easily get (\ref{eq:OptNetCons}). Finally, from (\ref{eq:thresholds}), we can re-formulate (\ref{eq:OptNetCons}) as
\begin{align}
        Z_0^{\psi^\ast}(\bm{\pi}^\ast(\bm{g})) &= \begin{cases}
\sigma_1(\bm{g}) - G_0 &, G_0 < \sigma_1(\bm{g}) \\ 
0 &, G_0 \in [\sigma_1(\bm{g}),\sigma_2(\bm{g})]\\ 
\sigma_2(\bm{g}) - G_0 & ,G_0 > \sigma_2(\bm{g}),
\end{cases}\nn
\end{align}
which proves the sign of $Z_0^{\psi^\ast}(\bm{\pi}^\ast(\bm{g}))$ under each piece.
\end{proof}

\subsection{Ex-Post Allocation}\label{subsec:ExpostAllocation}
Unlike the {\em ex-ante} price, the {\em ex-post} allocation is distributed after the prosumers schedule their DER. The operator may choose to accrue the {\em ex-post} charge amount of each prosumer to be distributed after multiple netting periods rather than at every netting period. The {\em ex-post} fee is essentially levied to achieve profit neutrality. After the price is announced and the transaction with the DSO is settled, the profit/deficit that the operator accumulates $A^\ast(\bm{g}):= \sum_{i \in \Bc} \sum_{n\in \Nc_i} A^\ast_n$ is, using Def.\ref{ax:ProfitNeutrality},
\begin{align}
    \sum_{i\in \Bc} \sum_{n\in \Nc_i} C^\chi_n(z_n) - C^{\mbox{\tiny NEM}}(\sum_{n\in \Nc} z_n)&=0 \stackrel{\text{(\ref{eq:Pcommunity}),(\ref{eq:PricingPolicy})}}{\Longrightarrow}\nn\\ \sum_{i \in \Bc} \sum_{n\in \Nc_i} (\pi^\ast_i(\bm{g}) \cdot z_n - A^\ast_n - \pi^{\mbox{\tiny NEM}}(Z_0) \cdot z_n)&=0 \nn\\
    \sum_{i \in \Bc} \sum_{n\in \Nc_i} (\pi^\ast_i(\bm{g}) \cdot z_n -\pi^{\mbox{\tiny NEM}}(Z_0) \cdot z_n)&=\sum_{i \in \Bc} \sum_{n\in \Nc_i} A^\ast_n\nn\\
    \sum_{i \in \Bc} \sum_{n\in \Nc_i} (\pi^\ast_i(\bm{g}) -\pi^{\mbox{\tiny NEM}}(Z_0) )\cdot z_n&= A^\ast(\bm{g}).\nn
\end{align}
One can see that the larger the differential between the energy sharing price and NEM price ($\pi^\ast_i(\bm{g})-\pi^{\mbox{\tiny NEM}}(Z_0), \forall i\in \Bc$), the larger the profit/deficit. Note that if the network constraints are non-binding, \ie $\overline{\eta}^\ast_i=\underline{\eta}^\ast_i=0, \forall i\in \Bc$, then $A^\ast(\bm{g})=0$, and the pricing policy becomes one-part; see D-NEM in \cite{Alahmed&Tong:24TEMPR}.

There might not be unique way to re-allocate the operator's profit/deficit $A^\ast(\bm{g})$. A profit-sharing coalitional game can be established to fairly re-allocate the operator's profit/deficit. In $\S$\ref{subsec:PriceUniformity}, we propose a proportional allocation rule that makes the payment function uniform for all prosumers.

\subsection{Stackelberg Equilibrium}\label{subsec:EquilibriumThm}
    Under the solution $(\chi^\ast, \psi^\ast)$, with $A^\ast(\bm{g})$ as in (\ref{subsec:ExpostAllocation}) ,the operator is profit-neutral.

We show next that the network-aware pricing achieves a Nash equilibrium to the leader-follower game in $\$$\ref{subsec:StackGame}.

\begin{theorem}\label{thm:Equilibrium}
    The solution ($\chi^\ast,\psi^\ast$) is a Stackelberg equilibrium that also achieves social optimality, \ie
    \begin{align} \label{eq:SurplusMax_model2}
(\bm{d}_1^{\psi^\ast},\ldots,\bm{d}_N^{\psi^\ast}) = \underset{(\bm{d}_1,\ldots,\bm{d}_N)}{\rm argmax}&~~ \sum_{i\in \Bc} \sum_{n\in \Nc_i} U_n(\bm{d}_n)-C^{\mbox{\tiny NEM}}(Z_0)\nn \\ \text{subject to} &~~ 
Z_0 = \sum_{n\in \Nc} \big(\bm{1}^\top \bm{d}_{n}-g_n\big) \nn\\&\hspace{0.4cm} \bm{d}_n \in \Dc_n ~ \forall n\in \Nc\nn\\
&\hspace{0.4cm} \underline{\bm{v}} \preceq  -\bm{R} \bm{Z} \preceq \overline{\bm{v}}.\nn
\end{align}
\end{theorem}
\begin{proof}
    See the appendix.
\end{proof}
The proof of Theorem \ref{thm:Equilibrium} solves an upper bound of (\ref{eq:community_optimization}) that relaxes the profit-neutrality constraint (\ref{eq:UpperLevelProfit}).

\subsection{Energy Sharing Payment Uniformity}\label{subsec:PriceUniformity}
We propose here a unique way to allocate the operator's profit/deficit $A^\ast(\bm{g})$. For every bus $i\in \Bc$, the allocation to every prosumer is given by
\begin{equation}\label{eq:FixedCharge}
    A_n^\ast(\bm{g}) =  \left( \pi^\ast_i(\bm{g}) - \pi^{\mbox{\tiny NEM}}(\sum_{n\in \Nc} z_n)\right) \cdot  z_n, \forall n\in \Nc_i,
\end{equation}
which has three favourable features. First, it redistributes the profit/deficit proportionally to the prosumers based on how far their energy-sharing price from the DSO's NEM price, which reflects how much they paid (got paid) for voltage correction. Second, it makes prosumer payment functions $C_n^{\chi^\ast} \forall n\in \Nc$ uniform. Indeed, plugging (\ref{eq:FixedCharge}) into (\ref{eq:PricingPolicy}) cancels out the locational dynamic price $\pi^\ast_i(\bm{g})$, and yields a simple, uniform payment function that charges customers based on the NEM price, \ie for every bus $i\in \Bc$,
$$C_n^{\chi^\ast}(z_n)= \pi^{\mbox{\tiny NEM}}(Z_0) \cdot z_n, \forall n\in \Nc_i.$$
Third, unlike the computationally expensive coalitional-game-based profit allocation schemes such as the Shapley value, the allocation rule in (\ref{eq:FixedCharge}) is straightforward and directly correlates the allocation to the energy-sharing price and the prosumer's net consumption. The decentralization argument may not hold under the allocation in (\ref{eq:FixedCharge}), as it compensates prosumers explicitly based on their own net consumption, which may influence their consumption decisions resulting in deviations from the welfare-maximizing decisions. It might be, however, too difficult for prosumers to anticipate if the operator performs the re-allocation at every multiple netting periods rather than at every single netting period.  


\section{Numerical Study}\label{sec:num}

\begin{figure}
    \centering
    \includegraphics[width=0.55\columnwidth]{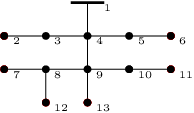}
    \caption{The IEEE 13-bus test feeder.}
    \label{fig:ieee13}
\end{figure}

\begin{figure}
    \centering
    \includegraphics[width=1.05\columnwidth]{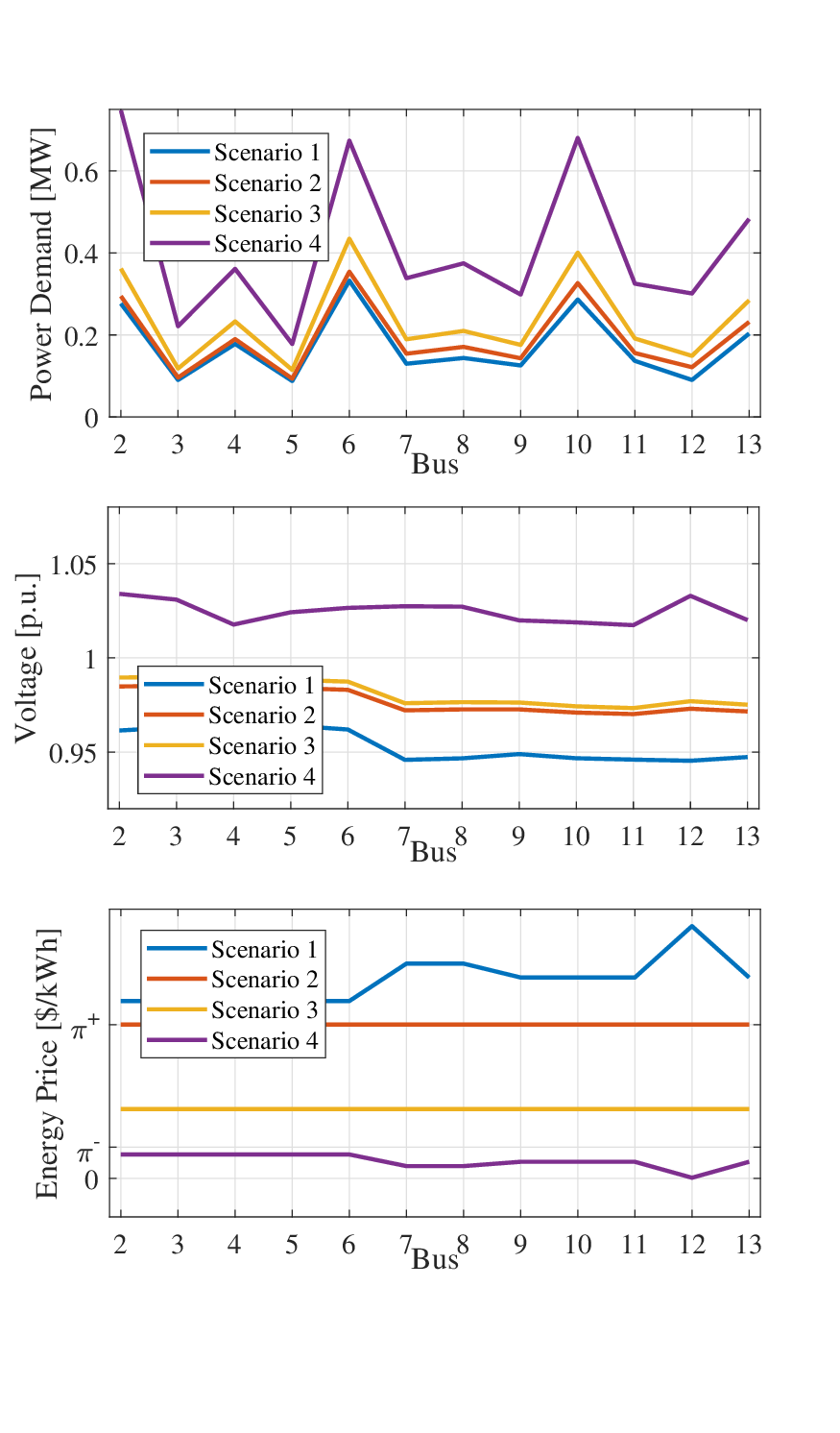}
    \vspace{-20mm}
    \caption{Summary of the numerical tests on the four considered scenarios. The lower panel reports the ex-ante energy prices obtained after solving the energy sharing platform optimization problem~\eqref{eq:community_optimization}. 
    The upper panel shows the cumulative power demand at each bus obtained after the energy sharing operator dispatched the energy prices.
    The middle panel reports the resulting bus voltage magnitudes.}
    \label{fig:results}
\end{figure}

Our network-aware market mechanism was validated on
the IEEE 13-bus feeder converted to a
single-phase equivalent~\cite{cavraro2019inverter}, see Figure~\ref{fig:ieee13}.
Bus 1 is the substation and represents the network slack bus.
Buses 2 to 13 instead host 23 prosumers. 
For every $n\in \Nc$, the following utility function was chosen
\begin{equation}\label{eq:UtilityForm}
   U_{n}(d_{n})=\left\{\begin{array}{ll}
\alpha_{n} d_{n}-\frac{1}{2}\beta_{n} d_{n}^2,\hspace{-0.2cm} &\hspace{-0.2cm} 0 \leq d_{n} \leq \frac{\alpha_{n}}{\beta_{n}} \\
\frac{\alpha_{n}^2}{2 \beta_{n}},\hspace{-0.2cm} &\hspace{-0.2cm} d_{n}>\frac{\alpha_{n}}{\beta_{n}},
\end{array} \right.
\end{equation}
where the parameters $\alpha_{n}, \beta_{n}$ were learned and calibrated using historical retail prices\footnote{The retail prices were taken from \href{https://data.austintexas.gov/stories/s/EOA-C-5-a-Austin-Energy-average-annual-system-rate/t4es-hvsj/}{Data.AustinTexas.gov} historical residential rates in Austin, TX.} and consumptions\footnote{For the historical consumption data, we used pre-2018 \href{https://www.pecanstreet.org/dataport/}{PecanStreet data} for households in Austin, TX.} and by assuming an elasticity of 0.21 taken from \cite{ASADINEJAD_Elasticity:18EPSR} (see appendix D in \cite{Alahmed&Tong:22EIRACM}). The minimum demand was set to ${\underline{\bm{d}}_n} = \bm{0}$ for every $n\in \Nc$, whereas the maximum demands ${\overline{\bm{d}}_n}$ and DER generations were obtained using data from the PecanStreet dataset.
We set $\bm{v}_{\text{min}} = 0.95$ p.u. and $\bm{v}_{\text{max}} = 1.05$ p.u.

In our simulations, we considered four scenarios, described in the following, that differ in the DER generation levels.
For each scenario, we used the exact AC power flow solver to obtain bus voltages, whereas we solved the optimization problems relying on the power flow equation linearization~\eqref{eq:LinearVoltages}.
The results are shown in Figure~\ref{fig:results}.

\emph{Scenario 1}: the DER generation here is zero for each prosumer. Hence, $G_0=0$ and $G_0 < \sigma_1(\bm{g})$. The energy-sharing system is importing energy. In this case, the energy sharing optimization problem solutions are such that $\underline{\eta}^\ast_i \neq 0$, \ie some voltages are on the lower bound $\bm{v}_{\text{min}}$.
The resulting prices are in general higher than $\pi^+$.

\emph{Scenario 2}: the DER generation $G_0$ is non-zero but still not enough to cover the demand, \ie $G_0 < \sigma_1(\bm{g})$. Hence, the energy-sharing system is importing energy. 
However, the optimum demands are such that all the voltages are within the desired bounds and the energy prices are equal $\pi^+$.

\emph{Scenario 3}: the DER generation was further increased in this scenario and $\sigma_1(\bm{g}) \leq G_0 \leq \sigma_2(\bm{g})$. That is, the energy-sharing platform did not exchange active power with the external network. The energy sharing platform optimization problem provides an energy price within $\pi^+$ and $\pi^-$; voltage limits are satisfied at the optimal consumption.

\emph{Scenario 4}: here, we increased the generation until $G_0 \geq \sigma_2(\bm{g})$. The platform exports power to the grid. The energy sharing optimization problem solution is such that the voltages in some locations are exactly $\bm{v}_{\text{max}}$ and the Lagrange multipliers vector $\overline{\bm{\eta}}^\ast$ is different from zero. The energy prices are smaller than $\pi^-$ and close to zero, \ie consumption is incentivized to take full advantage of generation.

Some observations are in order.
In general, we observe that increasing the DER generation $G_0$ results in the decrease of energy prices.
The energy prices can in principle be greater than $\pi^+$, see Scenario 1. This is to ensure that the voltage constraints are satisfied by decreasing the power demand.
Finally, we note a slight difference between the true and the expected (\ie the ones computed by the energy-sharing platform optimization problem) voltage magnitudes. Indeed, we see that the voltages in Scenario 4 are all strictly lower than $\bm{v}_{\text{max}}$ even though we obtained $\overline{\eta}^\ast_i \neq 0$, see the middle panel of Figure~\ref{fig:results}.
This can be explained by the fact that~\eqref{eq:community_optimization} was solved relying on the linearized equations~\eqref{eq:LinearVoltages} rather than on the true power flow equations. Note, however, that using the true equation would result in a nonconvex energy sharing optimization problem possibly displaying multiple local minima.

\section{Conclusion}\label{sec:conclusion}
In this work, we propose a network-aware and welfare-maximizing market mechanism for energy-sharing coalitions that aggregate small but ubiquitous BTM DER downstream of a DSO's revenue meter, charging the energy-sharing systems using a generic NEM tariff. The proposed pricing policy has {\em ex-ante} and {\em ex-post} pricing components. The {\em ex-ante} locational and threshold-based price decreases as the energy-sharing generation-to-demand ratio increases. The price is used to induce a collective prosumer reaction that decentrally maximizes social welfare while being network-cognizant. On the other hand, the {\em ex-post} charge/reward is used to enforce the market operator's profit-neutrality condition. We show that the market mechanism achieves an equilibrium to the Stackelberg game between the operator and its prosumers. We also show that a unique proportional rule to re-allocate the operator's profit/deficit can make the payment function of all energy-sharing prosumers {\em uniform}, even when the network constraints are binding.
Our simulation results leverage real DER data on an IEEE 13-bus test feeder system to show how the dynamic pricing drives the energy sharing's flexible consumption to abide by the network voltage limits.

\section{Acknowledgment}
The work of Ahmed S. Alahmed and Lang Tong was supported in part by the National Science Foundation under Award 2218110 and the Power Systems Engineering Research Center (PSERC) under Research Project M-46.
The work was authored in part by the National Renewable Energy Laboratory, operated by Alliance for Sustainable Energy, LLC, for the U.S. Department of Energy (DOE) under Contract No. DE-AC36-08GO28308. Funding is provided by the U.S. DOE Office of Energy Efficiency and Renewable Energy Building Technologies Office, United States. The views expressed in the article do not necessarily represent the views of the DOE or the U.S. Government. The U.S. Government retains and the publisher, by accepting the article for publication, acknowledges that the U.S. Government retains a nonexclusive, paid-up, irrevocable, worldwide license to publish or reproduce the published form of this work, or allow others to do so, for U.S. Government purposes.


\bibliographystyle{IEEEtran}
\bibliography{CDC}

\appendix[Incorporating Operating Envelopes]\label{app:generalizations}
Here, we present the pricing policy under OEs at the prosumer's revenue meter, as shown in Fig.\ref{fig:NetworkEC}.

OEs limit the net consumption of every prosumer $n\in \Nc$, as
\begin{equation}\label{eq:netconslimit}
z_n \in \mathcal Z_n := [\underline{z}_n,\overline{z}_n],
\end{equation}
where $\underline{z}_n\leq 0$ and $\overline{z}_n\geq 0$ are the export and import envelopes at the prosumers' meters, respectively. From the analysis in \cite{Alahmed&Cavraro&Bernstein&Tong:23AllertonArXiv}, the network-aware pricing policy generalizes as in the following policy.

\begin{policy}
For every bus $i\in \Bc$, the pricing policy charges the prosumers based on a two-part pricing
\begin{equation*}
    \chi^{\ast}: \bm{g} \mapsto C^{\chi^\ast}_n(z_n)=\underbrace{\pi^\ast_i(\bm{g})}_\text{ex-ante price}\cdot z_n - \underbrace{A_n^\ast}_\text{ex-post allocation}, \forall n\in \Nc_i,
\end{equation*}
where the {\em ex-ante} bus price $\pi^\ast_i(\bm{g})$ abides by a two-threshold policy with thresholds
\begin{equation}
\begin{aligned}
    \sigma_1(\bm{g}) &= \sum_{i\in \Bc} \sum_{n \in \Nc_i} \left[\bm{1}^{\top} {[\bm{f}_n(\bm{1}\chi_i^+)]}_{\underline{\bm{d}}_n}^{ \overline{\bm{d}}_n}\right]^{\overline{z}_n+g_n}_{\underline{z}_n+g_n},\nn\\ \sigma_2(\bm{g}) &= \sum_{i\in \Bc} \sum_{n \in \Nc_i} \left[\bm{1}^{\top} {[\bm{f}_n(\bm{1}\chi_i^-)]}_{\underline{\bm{d}}_n}^{ \overline{\bm{d}}_n}\right]^{\overline{z}_n+g_n}_{\underline{z}_n+g_n}\geq \sigma_1(\bm{g}),
\end{aligned}
\end{equation}
as
\begin{equation}
    \pi^\ast_i(\bm{g}) = \begin{cases} \chi^+_i(\bm{g}) & , G_0< \sigma_1(\bm{g}) \\ \chi^z_i(\bm{g}) & , G_0\in [\sigma_1(\bm{g}),\sigma_2(\bm{g})]\\ \chi^-_i(\bm{g}) & , G_0> \sigma_2(\bm{g}), \end{cases}\quad
\end{equation}
and the price $\chi_i^{\kappa}$, where $\kappa:=\{+,-,z\}$, is given by
\begin{equation}
    \chi_i^{\kappa} = \pi^\kappa-\sum_{j\in \Bc} R_{ji} (\overline{\eta}^\ast_j-\underline{\eta}^\ast_j)
\end{equation}
where $\overline{\eta}^\ast_j$ and $\underline{\eta}^\ast_j$ are the dual variables of the upper and lower voltage limits in (\ref{eq:UpperLevelNetwork}), respectively, and the price $\pi^z:=\mu^\ast$ is the solution of
\begin{equation*}
    \sum_{i\in \Bc} \sum_{n\in \Nc_i} \left[\bm{1}^\top [\bm{f}_n(\bm{1}\mu-\bm{1}\sum_{j\in \Bc} R_{ji}(\overline{\eta}^\ast_j-\underline{\eta}^\ast_j))]_{\underline{\bm{d}}_n}^{\overline{\bm{d}}_n}\right]^{\overline{z}_n+g_n}_{\underline{z}_n+g_n}= G_0.
\end{equation*}

For every bus $i\in \Bc$, the prosumer's ex-post charge/reward is denoted by $ A_n^\ast$, which we delineate in $\S$\ref{subsec:ExpostAllocation} and $\S$\ref{subsec:PriceUniformity}.
\end{policy}

\appendix[Mathematical Proofs]\label{app:proofs}
\subsection{Lemma \ref{lem:Central} and Proof of Lemma \ref{lem:Central}}
\begin{lemma}[Maximum welfare under centralized operation]\label{lem:Central}
    The maximum welfare under centralized operation, that solves
    \begin{align}
W^\ast(\bm{g}) := \underset{(\bm{d}_1,\ldots,\bm{d}_N)}{\rm maximize}&~~ \sum_{i\in \Bc} \sum_{n\in \Nc_i} U_n(\bm{d}_n)-C^{\mbox{\tiny NEM}}(Z_0)\nn \\ \text{subject to} &~~
Z_0 = \sum_{n\in \Nc} \big(\bm{1}^\top \bm{d}_{n}-g_n\big) \label{eq:centralOpt}\\&\hspace{0.4cm} \bm{d}_n \in \Dc_n ~ \forall n\in \Nc\nn\\
(\underline{\bm{\eta}},\overline{\bm{\eta}}):&\hspace{0.4cm} \underline{\bm{v}} \preceq  -\bm{R} \bm{Z} \preceq \overline{\bm{v}}.\nn
\end{align}
obeys by two-thresholds as
\begin{align*}
    W^\ast(\bm{g}) &= \sum_{i\in \Bc} \sum_{n\in \Nc_i}\nn\\&\begin{cases}  U_n(\bm{d}^+_n(\bm{g}))-\pi^+(\bm{1}^\top \bm{d}^+_n(\bm{g})-g_n)  & , G_0< \sigma_1(\bm{g}) \\ U_n(\bm{d}^z_n(\bm{g})) &\hspace{-1.2cm} , G_0\in [\sigma_1(\bm{g}),\sigma_2(\bm{g})]\\ U_n(\bm{d}^-_n(\bm{g}))-\pi^-(\bm{1}^\top \bm{d}^-_n(\bm{g})-g_n) & , G_0> \sigma_2(\bm{g}), \end{cases}
\end{align*}
where the thresholds ($\sigma_1(\bm{g}),\sigma_2(\bm{g})$) are as in (\ref{eq:thresholds}), the consumption $\bm{d}^\kappa_n(\bm{g}), \forall n\in \Nc$ and $\kappa:=\{+,-,z\}$, is 
\begin{align*}
    \bm{d}^\kappa_n(\bm{g}):= [\bm{f}_n(\bm{1}\pi^\kappa - \bm{1}\sum_{j \in \Bc}R_{ji}(\overline{\eta}^\ast_j-\underline{\eta}^\ast_j))]_{\underline{\bm{d}}_{n}}^{\overline{\bm{d}}_{n}}, \forall n\in \Nc_i,
\end{align*}
and $\overline{\eta}^\ast_i,\underline{\eta}^\ast_i\geq 0, \forall i\in \Bc$ is from the KKT conditions in (\ref{eq:KKTvoltageMax})-(\ref{eq:KKTvoltageMin}).
\end{lemma}
\subsection*{Proof of Lemma \ref{lem:Central}}
The convex non-differentiable program in (\ref{eq:centralOpt}) is a generalization to the standalone consumer decision problem under the DSO’s NEM X regime in \cite{Alahmed&Tong:22IEEETSG} with (a) the additional dimension of $N$ users located at $B$ buses and (b) network constraints. Therefore, (\ref{eq:centralOpt}) can be divided into three convex and differentiable programs based on the energy-sharing net consumption $Z_0$, namely ($\Pc_{Z_0\geq 0},\Pc_{Z_0= 0},\Pc_{Z_0\leq 0}$), as
\begin{align*}
\Pc_{Z_0\geq 0} : \underset{(\bm{d}_1,\ldots,\bm{d}_N)}{\rm maximize}&~~ \sum_{i\in \Bc} \sum_{n\in \Nc_i} U_n(\bm{d}_n)-\pi^+\cdot Z_0\nn \\ \text{subject to} &~~
Z_0\geq 0 \\&\hspace{0.4cm} \bm{d}_n \in \Dc_n ~ \forall n\in \Nc\nn\\
(\underline{\bm{\eta}}^+,\overline{\bm{\eta}}^+):&\hspace{0.4cm} \underline{\bm{v}} \preceq  -\bm{R} \bm{Z} \preceq \overline{\bm{v}}.\nn
\end{align*}
\begin{align*}
\Pc_{Z_0= 0} : \underset{(\bm{d}_1,\ldots,\bm{d}_N)}{\rm maximize}&~~ \sum_{i\in \Bc} \sum_{n\in \Nc_i} U_n(\bm{d}_n)\hspace{1.53cm}\nn \\ \text{subject to} &~~
Z_0= 0 \\&\hspace{0.4cm} \bm{d}_n \in \Dc_n ~ \forall n\in \Nc\nn\\
(\underline{\bm{\eta}}^-,\overline{\bm{\eta}}^-):&\hspace{0.4cm} \underline{\bm{v}} \preceq  -\bm{R} \bm{Z} \preceq \overline{\bm{v}}.\nn
\end{align*}
\begin{align*}
\Pc_{Z_0\leq 0} : \underset{(\bm{d}_1,\ldots,\bm{d}_N)}{\rm maximize}&~~ \sum_{i\in \Bc} \sum_{n\in \Nc_i} U_n(\bm{d}_n)-\pi^-\cdot Z_0\nn \\ \text{subject to} &~~
Z_0\leq 0 \\&\hspace{0.4cm} \bm{d}_n \in \Dc_n ~ \forall n\in \Nc\nn\\
(\underline{\bm{\eta}}^0,\overline{\bm{\eta}}^0):&\hspace{0.4cm} \underline{\bm{v}} \preceq  -\bm{R} \bm{Z} \preceq \overline{\bm{v}}.\nn
\end{align*}
It has been shown in \cite{Alahmed&Tong:22IEEETSG} that the optimal consumption policy is a two-thresholds policy on the aggregate renewables $G_0$. The three problems above are generalizations of \cite{Alahmed&Tong:22IEEETSG} that incorporate multiple prosumers and buses dimension and network constraints dimension, which therefore yields, for every bus $i\in \Bc$, the following optimal consumption vector
\begin{equation*}
    \bm{d}^\ast_n(\bm{g}) = \begin{cases} \bm{d}^+_n(\bm{g}) & , G_0< \sigma_1(\bm{g}) \\ \bm{d}^z_n(\bm{g}) & , G_0\in [\sigma_1(\bm{g}),\sigma_2(\bm{g})]\\ \bm{d}^-_n(\bm{g}) & , G_0> \sigma_2(\bm{g}), \end{cases},\quad \forall n \in \Nc_i,
\end{equation*}
where the thresholds ($\sigma_1(\bm{g}),\sigma_2(\bm{g})$) are as in (\ref{eq:thresholds}) and the vector $\bm{d}^\kappa_n(\bm{g})$, where $\kappa:=\{+,-,z\}$, is given by
\begin{align}\label{eq:CentralCons}
    \bm{d}^\kappa_n(\bm{g}):= [\bm{f}_n(\bm{1}\pi^\kappa - \bm{1}\sum_{j \in \Bc}R_{ji}(\overline{\eta}^\ast_j-\underline{\eta}^\ast_j))]_{\underline{\bm{d}}_{n}}^{\overline{\bm{d}}_{n}}, \forall n\in \Nc_i.
\end{align}
By generalizing the special case in \cite{Alahmed&Tong:22IEEETSG} through incorporating the additional dimension of $N$ users located at $B$ buses and the network constraints, one can see that the price $\pi^z(\bm{g})$ is as in (\ref{eq:NZprice}). Lastly, for every $i\in \Bc$, the dual variable $\overline{\eta}^\ast_i \geq 0$ is computed from the KKT conditions
\begin{subequations}\label{eq:KKTvoltageMax}
\begin{align}
    \overline{\eta}^\ast_i\left(- \sum_{j\in \Bc} R_{ij} \sum_{k \in \Nc_j} (\bm{1}^\top \bm{d}^\kappa_n(\bm{g}) - g_n) \right)&=0\\ - \sum_{j\in \Bc} R_{ij} \sum_{k \in \Nc_j} (\bm{1}^\top \bm{d}^\kappa_n(\bm{g}) - g_n) - \overline{v}_i & \leq 0
\end{align}
\end{subequations}
and $\underline{\eta}^\ast_i \geq 0$ is similarly computed from
\begin{subequations}\label{eq:KKTvoltageMin}
\begin{align}
    \underline{\eta}^\ast_i\left(- \sum_{j\in \Bc} R_{ij} \sum_{k \in \Nc_j} (\bm{1}^\top \bm{d}^\kappa_n(\bm{g}) - g_n) \right)&=0\\ - \sum_{j\in \Bc} R_{ij} \sum_{k \in \Nc_j} (\bm{1}^\top \bm{d}^\kappa_n(\bm{g}) - g_n) - \underline{v}_i & \geq 0
\end{align}
\end{subequations}
\hfill$\blacksquare$
\subsection{Proof of Theorem \ref{thm:Equilibrium}}
Recall the bi-level program of operator and prosumers decisions
\begin{subequations}\label{eq:community_optimizationProof}
\begin{align}
 \underset{\bm{C}^\chi(\cdot)}{\operatorname{maximize}}& \Bigg(W^{\chi_{\bm{\psi}}} = \sum_{n\in \Nc} U_n(\bm{d}^{\psi^\sharp_\chi}_{n}) - C^{\mbox{\tiny NEM}}(Z_0^{\psi^\sharp_\chi})\Bigg)\label{eq:UpperLevelObjProof}\\
 \text{subject to}&\quad \quad \sum_{n\in \Nc} C^\chi_n(z_n^{\psi^\sharp_\chi}) = C^{\mbox{\tiny NEM}}(Z_0^{\psi^\sharp_\chi})\label{eq:UpperLevelProfitProof}\\&
 \quad \quad Z_0^{\psi^\sharp_\chi} = \sum_{n\in \Nc} \big(\bm{1}^\top \bm{d}^{\psi^\sharp_\chi}_{n}- g_n\big)\label{eq:UpperLevelEnergyBalanceProof}\\
 (\underline{\bm{\eta}},\overline{\bm{\eta}})&\quad \quad \underline{\bm{v}} \preceq  -\bm{R} \bm{Z}^{\psi^\sharp_\chi} \preceq \overline{\bm{v}}\label{eq:UpperLevelNetworkProof}\\
 &\quad \quad \text{for all } i=1,\ldots, B, n \in \Nc_i \\
 &	\bm{d}^{\psi^\sharp_\chi}_{n}\hspace{-0.1cm} := \underset{\bm{d}_n \in \Dc_n}{\operatorname{argmax}}~S^\chi_n(\bm{d}_n,g_n)\hspace{-0.1cm}:=U_n(\bm{d}_n)-C^\chi_n(z_n)\label{eq:LowerLevelObjProof}\\
 &	\hspace{1.3cm}	\text{subject to}~~ z_n= \bm{1}^\top \bm{d}_n - g_n.\label{eq:LowerLevelEnergyBalanceProof}
\end{align}
\end{subequations}
We solve a relaxed version of (\ref{eq:community_optimizationProof}) that does not require profit-neutrality, hence (\ref{eq:UpperLevelProfitProof}) is removed, and for every $n\in \Nc$, we want to find the price $\pi_n$ in $\tilde{C}_n^\chi(\cdot) = \pi_n \cdot z_n$. So, (\ref{eq:community_optimizationProof}) is reformulated to
\begin{subequations}
\begin{align}
 \underset{\tilde{\bm{C}}^\chi(\cdot)}{\operatorname{maximize}}& \Bigg(W^{\chi_{\bm{\psi}}} = \sum_{n\in \Nc} U_n(\bm{d}^{\psi^\sharp_\chi}_{n}) - C^{\mbox{\tiny NEM}}(Z_0^{\psi^\sharp_\chi})\Bigg)\nn\\
 \text{subject to}&
 \quad \quad Z_0^{\psi^\sharp_\chi} = \sum_{n\in \Nc} \big(\bm{1}^\top \bm{d}^{\psi^\sharp_\chi}_{n}- g_n\big)\nn\\
 (\underline{\bm{\eta}},\overline{\bm{\eta}})&\quad \quad \underline{\bm{v}} \preceq  -\bm{R} \bm{Z}^{\psi^\sharp_\chi} \preceq \overline{\bm{v}}\nn\\
 &\quad \quad \text{for all } i=1,\ldots, B, n \in \Nc_i\nn \\
 &~~~ L_{kn}(d_{kn}^{\psi_\chi^\sharp})-\pi+ \overline{\lambda}^\ast_{kn} -  \underline{\lambda}^\ast_{kn}=0,~~ \forall k\in \Kc \nn\\
 &\hspace{1.6cm} \overline{d}_{kn} - d_{kn}^{\psi_\chi^\sharp} \geq 0 \perp \overline{\lambda}^\ast_{kn} \geq 0,~~ \forall k\in \Kc \nn\\
 &\hspace{2.35cm} d_{kn}^{\psi_\chi^\sharp} \geq  0 \perp \underline{\lambda}^\ast_{kn} \geq 0, ~~ \forall k\in \Kc\nn
\end{align}
\end{subequations}
where $x \perp y$ means that $x$ and $y$ are perpendicular. From Lemma \ref{lem:OptSchedule}, the lower level's KKT conditions can be replaced as the following
\begin{subequations}
\begin{align}
 \underset{\tilde{\bm{C}}^\chi(\cdot)}{\operatorname{maximize}}& \Bigg(W^{\chi_{\bm{\psi}}} = \sum_{n\in \Nc} U_n(\bm{d}^{\psi^\sharp_\chi}_{n}) - C^{\mbox{\tiny NEM}}(Z_0^{\psi^\sharp_\chi})\Bigg)\nn\\
 \text{subject to}&
 \quad \quad Z_0^{\psi^\sharp_\chi} = \sum_{n\in \Nc} \big(\bm{1}^\top \bm{d}^{\psi^\sharp_\chi}_{n}- g_n\big)\nn\\
 (\underline{\bm{\eta}},\overline{\bm{\eta}})&\quad \quad \underline{\bm{v}} \preceq  -\bm{R} \bm{Z}^{\psi^\sharp_\chi} \preceq \overline{\bm{v}}\nn\\
 &\quad \quad \text{for all } i=1,\ldots, B, n \in \Nc_i\nn \\
 &~~~ \bm{d}^{\psi^\sharp_\chi}_{n}(\pi_n) = [\bm{f}_{n}(\bm{1}\pi_n)]_{\underline{\bm{d}}_{n}}^{\overline{\bm{d}}_{n}}.\nn
\end{align}
\end{subequations}

Note that if $\pi^\ast_n$ is found, we have equilibrium. The program above is similar to the central program (\ref{eq:centralOpt}) in Lemma \ref{lem:Central}, but with the prices $\pi_n, \forall n\in \Nc$ as decision variables rather than the consumptions $\bm{d}_n, \forall n\in \Nc$. Therefore, the threshold structure in Lemma \ref{lem:Central} holds with the prices used to compute $\bm{d}^\ast_n, \forall n\in \Nc$ being the optimal prices. Therefore, we have, for every bus $i\in \Bc$,
\begin{equation*}
    \pi^\ast_i(\bm{g}) = \begin{cases} \pi^+-\sum_{j\in \Bc} R_{ji} (\overline{\eta}^\ast_j-\underline{\eta}^\ast_j) &\hspace{-0.3cm} , G_0< \sigma_1(\bm{g}) \\ \pi^z(\bm{g})-\sum_{j\in \Bc} R_{ji} (\overline{\eta}^\ast_j-\underline{\eta}^\ast_j) &\hspace{-0.3cm} , G_0\in [\sigma_1(\bm{g}),\sigma_2(\bm{g})]\\ \pi^--\sum_{j\in \Bc} R_{ji} (\overline{\eta}^\ast_j-\underline{\eta}^\ast_j) &\hspace{-0.3cm} , G_0> \sigma_2(\bm{g}), \end{cases}\quad
\end{equation*}
One can easily see from Lemma \ref{lem:OptSchedule}, that by plugging the equilibrium price $\pi^\ast_i(\bm{g}), \forall i\in \Bc$, the prosumers' optimal consumption decisions match the consumption decisions under centralized operation in Lemma \ref{lem:Central}.  

Now, we note that for every bus $i\in \Bc$, the pricing policy $\tilde{C}^{\chi^\ast}_n = \pi^\ast_i(\bm{g})\cdot z_n, \forall n\in \Nc_i$ does not achieve profit neutrality. Therefore, for every bus $i\in \Bc$, the pricing policy is augmented by the {\em ex-post} allocation $A^\ast_n$ to become $C^{\chi^\ast}_n = \pi^\ast_i(\bm{g})\cdot z_n - A^\ast_n, \forall n\in \Nc_i$ with $\sum_{i \in \Bc} \sum_{n\in \Nc_i} A^\ast_n=\sum_{i \in \Bc} \sum_{n\in \Nc_i} (\pi^\ast_i(\bm{g}) -\pi^{\mbox{\tiny NEM}}(Z_0) )\cdot z_n$. One can then apply Def.\ref{ax:ProfitNeutrality} to show that $\bm{C}^{\chi^\ast}$ achieves profit neutrality. \hfill$\blacksquare$
\end{document}